\documentclass[12pt]{article}
\usepackage{amsmath,amssymb}

\title{Inclusion Logic and Fixed Point Logic}
\author{Pietro Galliani \\ {\normalsize University of Helsinki} \\ {\normalsize pgallian@gmail.com} \and Lauri Hella \\ {\normalsize University of Tampere} \\ {\normalsize lauri.hella@uta.fi}}

\newtheorem{theorem}{Theorem}
\newtheorem{lemma}[theorem]{Lemma}
\newtheorem{corollary}[theorem]{Corollary}
\newtheorem{definition}[theorem]{Definition}
\newtheorem{proposition}[theorem]{Proposition}
\newtheorem{example}[theorem]{Example}

\newcommand{\dom}{\mathtt{Dom}}
\newcommand{\rank}{\mathtt{rk}}
\newcommand{\parts}{\mathcal P}
\newcommand{\tuple}{\vec}
\newcommand{\ar}{\mathtt{ar}}
\newcommand{\indep}[2]{{#1}\bot{#2}}
\newcommand{\indepr}[3]{{#2}\bot_{#1}{#3}}

\newcommand{\gfp}{\mathsf{gfp}}
\newcommand{\lfp}{\mathsf{lfp}}

\newcommand{\INCL}{\mathrm{FO}(\subseteq)}
\newcommand{\DL}{\mathrm{FO(D)}}
\newcommand{\EXCL}{\mathrm{FO}(\,|\,)}
\newcommand{\INDL}{\mathrm{FO}(\bot)}

\newcommand{\CINDL}{\mathrm{FO}(\bot_c)}
\newcommand{\IEL}{\mathrm{FO}(\subseteq, \,|\,)}
\newcommand{\GFP}{\mathrm{GFP}}
\newcommand{\LFP}{\mathrm{LFP}}
\newcommand{\PGFP}{\mathrm{GFP^+}}

\newcommand{\PTIME}{\mathsf{PTIME}}
\newcommand{\NP}{\mathsf{NP}}

\newcommand{\CFMX}{\mathcal{C}(X,\dom(M))}
\newcommand{\CFAX}{\mathcal{C}(X,A)}
\newcommand{\CFBY}{\mathcal{C}(Y,B)}

\newcommand{\CFMXn}{\mathcal{C}(X,\dom(M)^n)}

\begin{document}
\maketitle

\begin{abstract}
We investigate the properties of Inclusion Logic, that is, First Order Logic with Team Semantics extended with inclusion dependencies.  We prove that Inclusion Logic is equivalent to Greatest Fixed Point Logic, 
and we prove that all union-closed first-order definable properties of relations are definable in it. 
We also provide an Ehrenfeucht-Fra\"iss\'e game for Inclusion Logic, and give an example illustrating
its use.
\end{abstract}
\section{Introduction}


Inclusion Logic \cite{galliani12}, $\INCL$, is a novel logical formalism designed 
for expressing inclusion dependencies between variables. It is closely related to Dependence Logic
\cite{vaananen07}, $\DL$, which is the extension of First Order Logic by functional dependencies between
variables. Dependence Logic initially arose as a variant of \emph{Branching Quantifier Logic} \cite{henkin61} and of \emph{Independence-Friendly Logic} \cite{hintikkasandu89,mann11}, and its study has sparked the development of a whole family of logics obtained by adding various dependency conditions into First Order Logic.

All these logics are based on Team Semantics \cite{hodges97,vaananen07} which is a 
generalization of Tarski Semantics. In Team Semantics, formulas are satisfied or not satisfied by \emph{sets} of assignments, called \emph{teams}, rather than by single assignments. This semantics was introduced in \cite{hodges97} for the purpose of defining a compositional equivalent for the Game Theoretic Semantics of Independence-Friendly Logic \cite{hintikkasandu89,mann11}, but it was soon found out to be of independent interest. See \cite{galliani12c} for a, mostly up-to-date, account of the research on Team Semantics.

Like Branching Quantifier Logic and Independence-Friendly Logic, Dependence Logic has
the same expressive power as Existential Second Order Logic $\Sigma^1_1$: every 
$\DL$-sentence  is equivalent to some $\Sigma^1_1$-sentence, and vice versa 
\cite{vaananen07}. The semantics of Dependence Logic is downwards closed in the sense that
if a team $X$ satisfies a formula $\phi$ in a model $M$, then all subteams $Y\subseteq X$
also satisfy $\phi$ in $M$. The equivalence between $\DL$ and $\Sigma^1_1$
was extended to formulas in \cite{kontinenv09}, where it was proved that $\DL$
captures exactly the downwards closed $\Sigma^1_1$-definable properties of teams.

Other variants of Dependence Logic that have been studied are Conditional Independence Logic 
$\CINDL$ \cite{gradel10}, Independence Logic $\INDL$ \cite{gradel10,vaananen13}, Exclusion Logic 
$\EXCL$ \cite{galliani12} and  Inclusion/Exclusion Logic $\IEL$ \cite{galliani12}. All the logics 
in this family arise from dependency notions that have been studied in Database Theory. In particular, 
$\DL$ is based on \emph{functional dependencies} introduced by Armstrong \cite{armstrong74}, 
$\INCL$ is based on \emph{inclusion dependencies} \cite{fagin81,casanova82}, 
$\EXCL$ is based on \emph{exclusion dependencies} \cite{casanova83}, and 
$\INDL$ is based on \emph{independence conditions} \cite{geiger91}.

The expressive power of all these logics, with the exception of $\INCL$, is well understood. It is known
that, with respect to sentences, they are all equivalent with $\Sigma^1_1$. With respect to formulas, 
$\EXCL$ is equivalent with $\DL$ \cite{galliani12}; and $\IEL$, $\CINDL$ and $\INDL$ are all 
equivalent to each other \cite{galliani12,vaananen13}. Moreover, $\CINDL$ (and hence also
$\IEL$ and $\INDL$) captures all $\Sigma^1_1$-definable properties of teams \cite{galliani12}.

On the other hand, relatively little is known about the expressive power of Inclusion Logic, and the main purpose of the present work is precisely to remedy this. What little is known about this formalism can be found in \cite{galliani12}, and amounts to the following: With respect to formulas, $\INCL$ is strictly weaker
than $\Sigma^1_1\equiv\CINDL$ and incomparable with $\DL\equiv\EXCL$. This is simply because
the semantics of $\INCL$ is not downwards closed, but is closed under unions: if both teams $X$ and
$Y$ satisfy a formula $\phi$ in a model $M$, then $X\cup Y$ also satisfies $\phi$ in $M$.
Moreover, it is known that $\INCL$ is stronger than First Order Logic over sentences, and that it is contained in $\Sigma^1_1$; but it was an open problem whether it it is equivalent to $\Sigma^1_1$,
or whether $\INCL$-formulas could define all union closed
$\Sigma^1_1$-definable properties of teams.

In this paper we show that the answer to both of these problems is negative. In fact, we give
a complete characterization for the expressive power of $\INCL$ in terms of 
Positive Greatest Fixed Point Logic $\PGFP$: We prove that every $\INCL$-sentence is
equivalent to some $\PGFP$-sentence, and vice versa (Corollary \ref{incgfpeq}).
Moreover, we prove that a property of teams is definable by an $\INCL$-formula
if and only if it is expressible by a $\PGFP$-formula in a specific way (Theorems \ref{inc-gfp}
and \ref{gfp-inc}). 

Fixed point logics have a central role in the area of Descriptive Complexity Theory. By the 
famous result of Immerman \cite{immerman86} and Vardi \cite{vardi82}, Least Fixed Point
Logic $\LFP$ captures $\PTIME$ on the class of ordered finite models. Furthermore, it is well known
that on finite models, $\LFP$ is equivalent to $\PGFP$. Thus, we obtain a novel characterization
for $\PTIME$: a class of ordered finite models is in $\PTIME$ if and only if it is definable
by a sentence of $\INCL$. 

In addition to the equivalence with $\PGFP$, we prove that all union-closed first-order
definable properties of teams are definable in Inclusion Logic (Corollary \ref{FOunion}). 
Thus, it is not possible to increase the expressive power of $\INCL$ by adding first-order
definable union-closed dependencies. On the other hand, it is an interesting open problem,
whether $\INCL$ can be extended by some natural set $\mathbf{D}$ of union-closed dependencies
such that the extension $\mathrm{FO}(\subseteq,\mathbf{D})$ captures all 
union-closed $\Sigma^1_1$-definable properties of teams.

We also introduce a new Ehrenfeucht-Fra\"iss\'e game that characterizes the expressive 
power of Inclusion Logic (Theorem \ref{efgthm}).
Our game is a modification of the EF game for Dependence Logic defined in \cite{vaananen07}.
Although the EF game has a clear second order flavour, it is still more manageable than 
the usual EF game for $\Sigma^1_1$; we illustrate this by describing a concrete 
winning strategy for Duplicator in the case of models with empty signature
(Proposition \ref{efex}). 
Due to the equivalence between $\INCL$ and $\PGFP$ we see that the EF game for 
Inclusion Logic is also a novel EF game for $\PGFP$; it is quite different 
in structure from the one introduced in \cite{bosse93}. It may be hoped that this new game 
and its variants could be of some use for studying the expressive power of fixed point logics.

\section{Preliminaries}

\subsection{Team Semantics}

In this section, we will recall the definition of the Team Semantics for First Order Logic. For simplicity reasons, we will assume that all our expressions are in negation normal form.

\begin{definition}
Let $M$ be a first order model and let $V$ be a set of variables. A \emph{team} $X$ over $M$ with 
\emph{domain} $\dom(X) = V$ is a set of assignments $s: V \rightarrow \dom(M)$. Given a tuple 
$\tuple t=(t_1,\ldots,t_n)$ of terms with variables in $V$ and an assignment $s\in X$, we write 
$\tuple t \langle s \rangle$ for the tuple $(t_1\langle s \rangle,\ldots,t_n\langle s \rangle)$, where 
$t\langle s \rangle$ denotes the value of the term $t$ with respect to $s$ in the model $M$.
Furthermore, we write
$X(\tuple t)$ for the relation $\{\tuple t \langle s \rangle : s \in X\}$.

A (non-deterministic) \emph{choice function} for a team $X$ over a set $A$ is a function $H: X\to
\parts(A)\setminus\{\emptyset\}$. The set of all choice functions for $X$ over $A$
is denoted by $\CFAX$.
\end{definition}

\begin{definition}[Team Semantics for First Order Logic\footnote{What we present here is the so-called \emph{lax} version of Team Semantics. There also exists a \emph{strict} version, with somewhat different rules for disjunction and existential quantification. As discussed in \cite{galliani12}, the lax semantics has more convenient properties for the case of 
Inclusion Logic.}]
Let $M$ be a first order model and let $X$ be a team over it. Then, for all first-order literals $\alpha$, variables $v$, and formulas $\phi$ and $\psi$ over the signature of $M$ and with free variables in $\dom(X)$, 
\begin{description}
\item[TS-lit:] $M \models_X \alpha$ iff for all $s \in X$, $M \models_s \alpha$ in the usual Tarski Semantics sense; 
\item[TS-$\vee$:] $M \models_X \phi \vee \psi$ iff $X = Y \cup Z$ for some $Y$ and $Z$ such that $M \models_Y \phi$ and $M \models_Z \psi$; 
\item[TS-$\wedge$:] $M \models_X \phi \wedge \psi$ iff $M \models_X \phi$ and $M \models_X \psi$; 
\item[TS-$\exists$:] $M \models_X \exists v \phi$ iff there exists a function $H\in\CFMX$ such that $M \models_{X[H/v]} \psi$, where $X[H/v] = \{s[m/v] : s \in X, m \in H(s)\}$;
\item[TS-$\forall$:] $M \models_X \forall v \phi$ iff $M \models_{X[M/v]} \phi$, where $X[M/v] = \{s[m/v] : s \in X, m \in \dom(M)\}$. 
\end{description}
\end{definition}

The next theorem can be proved by structural induction on $\phi$:
\begin{theorem}[Team Semantics and Tarski Semantics]
\label{thm:TS-FOL}
For all first order formulas $\phi(\tuple v)$, all models $M$ and all teams $X$, $M \models_X \phi$ if and only if for all $s \in X$, $M \models_s \phi$ with respect to Tarski Semantics.
\end{theorem}

Thus, in the case of First Order Logic it is possible to reduce Team Semantics to Tarski Semantics.
What is then the point of working with the technically more complicated Team Semantics? As we will see in the next subsection, the answer is that Team Semantics allows us to extend First Order Logic in novel and interesting ways. 

Note that on every model $M$, there are two teams with empty domain: the empty team $\emptyset$,
and the team $\{\emptyset\}$ containing the empty assignment $\emptyset$. All the logics 
that we consider in this paper have the \emph{empty team property}: 
$M\models_\emptyset\phi$ for every formula
$\phi$ and model $M$. Thus, we say that a \emph{sentence $\phi$ is true} in a model $M$ if 
$M\models_{\{\emptyset\}}\phi$. If this is the case, we drop the subscript $\{\emptyset\}$, and 
write just $M\models\phi$.

\subsection{Dependencies in Team Semantics}

As we saw, in Team Semantics formulas are satisfied or not satisfied by sets of assignments, called \emph{teams}; and a team corresponds in a natural way to a relation over the domain of the model. Therefore, any property of relations can be made to correspond to some property of teams, which we can then add to our language as a new atomic formula. In particular, we can do so for database-theoretic dependency notions, thus obtaining the following \emph{generalized atoms}:\footnote{The notion of ``generalized atom'' is defined formally in \cite{kuusisto12}. 
}

\begin{definition}[Dependence Atoms]
Let $\tuple t_1$, $\tuple t_2$, $\tuple t_3$ be tuples of terms over some vocabulary. Then, for all models $M$ and all teams $X$ over $M$ whose domain contains the variables of $\tuple t_1 \tuple t_2 \tuple t_3$, 
\begin{description}
\item[TS-fdep:] $M \models_X\; =\!\!(\tuple t_1, \tuple t_2)$ if and only if, for all $s, s' \in X$, $\tuple t_1\langle s \rangle = \tuple t_1\langle s'\rangle \Rightarrow \tuple t_2 \langle s\rangle = \tuple t_2\langle s'\rangle$; 
\item[TS-exc:] For $|\tuple t_1| = |\tuple t_2|$, $M \models_X \tuple t_1 \mid \tuple t_2$ if and only if $X(\tuple t_1) \cap X(\tuple t_2) = \emptyset$;
\item[TS-inc:] For $|\tuple t_1| = |\tuple t_2|$, $M \models_X \tuple t_1 \subseteq \tuple t_2$ if and only if $X(\tuple t_1) \subseteq X(\tuple t_2)$;
\item[TS-ind:] $M \models_X \indep{\tuple t_1}{\tuple t_2}$ if and only if for all $s, s' \in X$ there exists a $s'' \in X$ with $\tuple t_1\langle s''\rangle = \tuple t_1 \langle s \rangle$ and $\tuple t_2 \langle s''\rangle = \tuple t_2 \langle s'\rangle$;
\item[TS-cond-ind:]	$M \models_X \indepr{\tuple t_1}{\tuple t_2}{\tuple t_3}$ if and only if for all $s, s' \in X$ with $\tuple t_1\langle s\rangle = \tuple t_1\langle s'\rangle$ there exists a $s'' \in X$ with $(\tuple t_1 \tuple t_2)\langle s''\rangle = (\tuple t_1 \tuple t_2) \langle s\rangle$ and $(\tuple t_1 \tuple t_3)\langle s''\rangle = (\tuple t_1 \tuple t_3)\langle s'\rangle$.
\end{description}
\end{definition}

These atoms correspond respectively to \emph{functional dependencies} \cite{armstrong74}, to \emph{exclusion dependencies} \cite{casanova83}, to \emph{inclusion dependencies} \cite{fagin81,casanova82}, to \emph{independence conditions} \cite{geiger91}, and to \emph{conditional independence conditions}\footnote{As observed in \cite{engstrom11}, conditional independence atoms also correspond to \emph{embedded multivalued dependencies}.}; and by adding them to the language of First Order Logic we can obtain various logics, whose principal known properties we will now briefly recall.

\medskip

\textbf{Dependence Logic} $\DL$ is obtained by adding functional dependence atoms to the language of First Order Logic. It is the oldest and the most studied among the logics that we will discuss in this work, having been introduced in the seminal book \cite{vaananen07} as an alternative approach to the study of \emph{Branching} \cite{henkin61} and \emph{Independence-Friendly} \cite{hintikkasandu89,mann11} Quantification. It is \emph{downwards closed}, in the sense that, for all models $M$, Dependence Logic formulas $\phi$ and teams $X$, if $M \models_X \phi$ then $M \models_Y \phi$ for all subsets $Y$ of $X$. 

On the level of sentences, Dependence Logic has the same expressive power as Existential Second Order Logic $\Sigma_1^1$.
\begin{theorem}[\cite{walkoe70,enderton70,vaananen07}]
Every $\DL$-sentence is equivalent to some $\Sigma_1^1$-sentence, and vice versa.
In particular, $\DL$ captures $\NP$ on finite models.
\end{theorem}

The equivalence between $\DL$ and $\Sigma^1_1$ was extended to formulas by Kontinen and 
V\"a\"an\"anen, who  proved the following characterization:
\begin{theorem}[\cite{kontinenv09}]\label{konvaa}
Let $\phi$ be a $\DL$-formula with free variables in $\tuple v$. Then there exists a $\Sigma_1^1$-sentence 
$\Phi(R)$, where $R$ is a $|\tuple v|$-ary relation symbol which occurs only negatively in $\Phi$, such that 
\[
	M \models_X \phi \iff (M,X(\tuple v)) \models \Phi(R)\text{ for all models 
	$M$ and teams $X\not=\emptyset$.}
\]
Conversely, for any such $\Phi(R)$ there exists an $\DL$-formula $\phi$ such that the above holds.
\end{theorem}
Thus, $\DL$ is the strongest logic that can be obtained by adding $\Sigma_1^1$-definable downwards-closed dependence conditions to First-Order Logic. Indeed, any such condition will be expressible as $\exists S (X(\tuple v) \subseteq S \wedge \Phi(S))$ for some $\Phi$ in $\Sigma_1^1$, and therefore it will be equivalent to some $\DL$-formula. 
\medskip

\textbf{Exclusion Logic} $\EXCL$, on the other hand, is the logic obtained by adding exclusion atoms to First-Order Logic. It was introduced in \cite{galliani12}, where it was shown to be equivalent to Dependence Logic with respect to formulas.

\medskip

\textbf{Conditional Independence Logic} $\CINDL$, which was introduced in \cite{gradel10}, adds conditional independence atoms $\indepr{\tuple t_1}{\tuple t_2}{\tuple t_3}$ to the language of First Order Logic. Like $\DL$, $\CINDL$ is equivalent to $\Sigma_1^1$ with respect to sentences, and also with respect to formulas: 
\begin{theorem}[\cite{gradel10}]
Every $\CINDL$-sentence is equivalent to some $\Sigma_1^1$-sentence, and vice versa.
\end{theorem}
\begin{theorem}[\cite{galliani12}]
A class of relations is definable in Conditional Independence Logic if and only if it contains the empty relation and it is $\Sigma_1^1$-definable.
\end{theorem}
Therefore, Conditional Independence Logic is the strongest logic that can be obtained by adding $\Sigma_1^1$-definable dependencies which are true of the empty relation to First Order Logic. In particular, this implies that every $\DL$ formula (and, therefore, every $\EXCL$ formula) is equivalent to some $\CINDL$ formula.\footnote{This was already shown in \cite{gradel10}, in which it was shown that any dependence atom $=\!\!(\tuple t_1, \tuple t_2)$ is equivalent to the conditional independence atom $\indepr{\tuple t_1}{\tuple t_2}{\tuple t_2}$.} 
However, the converse is not true, since $\CINDL$ formulas are not, in general, downwards closed. 

Furthermore, \textbf{Inclusion/Exclusion Logic} $\IEL$ -- that is, the logic obtained by adding inclusion \emph{and} exclusion dependencies to First Order Logic -- was proved in \cite{galliani12} to be equivalent with $\CINDL$ with respect to formulas.

Finally, \textbf{Independence Logic} $\INDL$ is the logic obtained by adding only non-conditional dependence atoms $\indep{\tuple t_1}{\tuple t_2}$ to First Order Logic. As proved in \cite{vaananen13}, Independence Logic and Conditional Independence Logic are also equivalent with respect to formulas. 
\medskip

\textbf{Inclusion Logic} 
$\INCL$ is obtained by adding inclusion atoms to First Order Logic. It is not downwards closed, but it is \emph{closed under unions} in the following sense: if $\phi$ is an $\INCL$-formula,
$M$ is a model, and $X_i$, $i \in I$, are teams on $M$ such that $M \models_{X_i} \phi$ for all $i \in I$,
then $M \models_X \phi$, where $X=\bigcup_{i\in I} X_i$. (For a proof, see \cite{galliani12}).

Relatively little is known about the expressive power of $\INCL$, and the main purpose of the present work is precisely to remedy this. Here we only recall the following results from \cite{galliani12}:
\begin{enumerate}
\item On the level of formulas, $\INCL$
is strictly weaker than $\CINDL\equiv\INDL\equiv\Sigma^1_1$, and incomparable with $\DL\equiv\EXCL$.
\item The complement of the transitive closure of any first-order formula $\phi(\tuple x,\tuple y)$
is definable in $\INCL$; hence, $\INCL$ is strictly stronger than First Order Logic on sentences.
\item On the level of sentences, $\INCL$ is contained in $\Sigma^1_1$.
\end{enumerate} 

We give next a couple of further examples of the expressive power of $\INCL$.  

\begin{example}
(a) Consider the sentence $\phi:=\exists x\exists y(y\subseteq x\land Exy)$. 
Let $M=(\dom(M),E^M)$ be a finite model. Then
$M\models\phi$ if and only if $E^M$ contains a cycle, i.e., there are  $a_0,\ldots,a_{n-1}\in\dom(M)$
such that $(a_i,a_{i+1})\in E^M$ for all $i<n-1$, and $(a_{n-1},a_0)\in E^M$. 

The idea here is the following: by the lax semantics, the first existential quantifier gives a set $C$
of values for $x$, and the formula $\exists y(y\subseteq x\land Exy)$ then says that
for every $a\in C$ there is a $b\in C$ such that $(a,b)\in E^M$.

(b) Let $\psi$ be the  $\INCL$-sentence
$\exists w(\exists u(Pu\land u\subseteq w)\land
\forall u(Ewu\to\exists v(Euv\land v\subseteq w)))$.
Then $M\models\psi$ if and only if player I has a winning strategy in the following
game $G(M)$: Player I starts by choosing some element $a_0\in P^M$. In each odd round $i+1$,
player II chooses an element $a_{i+1}$ such that $(a_i,a_{i+1})\in E^M$. In each even round $i+1$,
player I chooses an element $a_{i+1}$ such that $(a_i,a_{i+1})\in E^M$. The first player unable to move
according to the rules, loses the game. Player I wins all infinite plays of the game. 

The class $K$ of all finite models $M$ such that player II has a winning strategy in $G(M)$
is an equivalent to Immerman's \emph{alternating graph accessibility problem}, $\mathrm{AGAP}$.
It is well known that $\mathrm{AGAP}$ is a complete problem for $\PTIME$ with respect
to quantifier free reductions (\cite{immerman87}). 
\end{example}

\subsection{Greatest Fixed Point Logic}

Let $ \psi(R, \tuple x)$ be a first-order formula such that the arity of $R$, $\ar(R)$, is equal 
to the length $k=|\tuple x|$ of the tuple $\tuple x$.
If $M$ is a model, then $\psi$ defines an operation $\Gamma=\Gamma_{M,\psi}$ on
the set $\parts(\dom(M)^k)$ of $k$-ary relations on $\dom(M)$ as follows: 
\[
	\Gamma(P):=\{\tuple a: (M,P)\models_{s[\tuple a/\tuple x]}\psi(R,\tuple x)\}
	\text{ for each }P\in\parts(\dom(M)^k).
\]
A relation $P$ is a \emph{fixed point} of the operation $\Gamma_{M,\psi}$ on $M$ if $\Gamma(P)=P$.
Furthermore, $P$ is the \emph{greatest fixed point} (\emph{least fixed point}) of $\Gamma_{M,\psi}$
if $Q\subseteq P$ ($P\subseteq Q$, respectively) for all fixed points $Q$ of $\Gamma_{M,\psi}$.

It is well known that if $R$ occurs only positively in $\psi$, then for every model $M$, 
$\Gamma_{M,\psi}$ has a greatest fixed point (as well as a least fixed point).
Moreover, the greatest fixed point $P$ of $\Gamma_{M,\psi}$ has the following characterization:
$P=\bigcup\{Q\subseteq\dom(M)^k: Q\subseteq \Gamma_{M,\psi}(Q)\}$ (see, e.g. \cite{libkin04}).

\begin{definition}
\emph{Greatest Fixpoint Logic}, $\GFP$, is obtained by adding to First Order Logic the 
\emph{greatest fixed point operator} $[\gfp_{R, \tuple x} \psi(R, \tuple x)] \tuple t$, where $R$ is a 
relation variable with $\ar(R)=|\tuple x|$, $\psi(R, \tuple x)$ is 
a formula in which $R$ occurs only positively, and $\tuple t$ is a tuple of terms with $|\tuple t|=|\tuple x|$. 
The semantics of the operator $\gfp$ is defined by the clause:
\begin{itemize}
\item $M\models_s [\gfp_{R, \tuple x} \psi(R, \tuple x)] \tuple t$  if and only if 
$\tuple t\langle s\rangle$ is in the greatest fixed point of $\Gamma_{M,\psi}$.
\end{itemize}

\emph{Positive Greatest Fixed Point Logic}, $\PGFP$, is the fragment of Greatest Fixed Point Logic in which fixed point operators occur only positively.

\emph{Least Fixpoint Logic}, $\LFP$, similarly, introduces an operator\\
$[\lfp_{R, \tuple x} \psi(R, \tuple x)] \tuple t$, again for $R$ occurring only positively in $\psi$, 
such that $M\models[\lfp_{R, \tuple x} \psi(R, \tuple x)] \tuple t$ if and only if 
$\tuple t \langle s\rangle$ is in the least fixed point of $\Gamma_{M,\psi}$.
\end{definition}

Fixed point logics have been the object of a vast amount of research, especially because of 
their applications in Finite Model Theory and Descriptive Complexity Theory.
In particular, Least Fixed Point Logic captures the complexity class $\PTIME$ that consists of
all problems that are solvable in polynomial time:

\begin{theorem}[\cite{immerman86,vardi82}]
A class of linearly ordered finite models is definable in $\LFP$ if and only if it can be recognized 
in $\PTIME$.
\end{theorem}

Another important result is that on finite models, Greatest Fixed Point Logic has the same 
expressive power as Least Fixed Point Logic. 

\begin{theorem}[\cite{immerman86}]
Over finite models, $\PGFP$ (as well as $\GFP$) is equivalent to $\LFP$.
\end{theorem}


We will also make use of the following normal form result for Positive Greatest Fixed Point Logic:

\begin{theorem}[\cite{moschovakis74,immerman86}]\label{gfp-nf}
Every $\PGFP$-sentence $\phi$ is equivalent to a $\PGFP$-sentence of the form
$\exists \tuple z\, [\gfp_{R,\tuple x}\,\psi(R,\tuple x)]\tuple z$, 
where $\psi$ is a first-order formula.
\end{theorem}

\section{Inclusion Logic captures $\PGFP$}
\label{sub:defIL}

We will now prove that Inclusion Logic has exactly the same expressive power as 
Positive Greatest Fixed Point Logic. Since the semantics of $\PGFP$ is defined
in terms of single assignments instead of teams, the equivalence of $\INCL$ and 
$\PGFP$ on formulas has to be formulated in a bit indirect way; see Theorems \ref{inc-gfp}
and \ref{gfp-inc} below.

We start with a lemma that connects teams and the greatest fixed point operator:

\begin{lemma}\label{gfp-lemma} Let $\psi(S,\tuple x)$ a $\PGFP$-formula with free variables in\\
$\tuple x =(x_1,\ldots,x_n)$
such that $S$ is $n$-ary and occurs only positively in $\psi$, let $M$ be a model, and let $Y$ a team on $M$.
\begin{description}
\item[(a)] If $(M,Y(\tuple x))\models_s \psi(S,\tuple x)$ for all $s\in Y$, then 
$M\models_s [\gfp_{S,\tuple x}\,\psi(S,\tuple x)]\tuple x$ for all $s\in Y$.

\item[(b)] If $Y$ is a maximal team such that $M\models_s [\gfp_{S,\tuple x}\,\psi(S,\tuple x)]\tuple x$ 
for all $s\in Y$, then $(M,Y(\tuple x))\models_s \psi(S,\tuple x)$ for all $s\in Y$.
\end{description}
\end{lemma}

\begin{proof}
Note that $(M,Y(\tuple x))\models_s \psi(S,\tuple x)$ for all $s\in Y$ if and only if 
$Y(\tuple x)\subseteq\Gamma_{M,\psi}(Y(\tuple x))$. 
Thus, claim (a) follows from the fact that the greatest fixed point of $\Gamma_{M,\psi}$
is the union of all relations $Q$ such that $Q\subseteq \Gamma_{M,\psi}(Q)$.
Claim (b) follows from the observation that if $Y$ is a maximal team such that 
$M\models_s [\gfp_{S,\tuple x}\,\psi(S,\tuple x)]\tuple x$ 
for all $s\in Y$, then $Y(\tuple x)$ is the greatest fixed point of $\Gamma_{M,\psi}$.
\end{proof}

We will next prove that every $\INCL$-formula can be expressed in $\PGFP$.

\begin{theorem}\label{inc-gfp}
For every $\INCL$-formula $\phi(\tuple x)$ with free variables in $\tuple x=(x_1,\ldots,x_n)$ there is a 
$\PGFP$-formula $\phi^*=\phi^*(R,\tuple x)$ such that $\ar(R)=|\tuple x|$, $R$ occurs only 
positively in $\phi^*$, and the condition
\[
	M \models_X \phi(\tuple x) \iff  (M,X(\tuple x)) \models_s \phi^*(R,\tuple x)\text{ for all $s \in X$} 
\]
holds for all models $M$ and teams $X$ with $\dom(X)=\{x_1,\ldots,x_n\}$. 
\end{theorem}
\begin{proof}
The proof is by structural induction on $\phi$.
\begin{enumerate}

\item If $\phi(\tuple x)$ is a first-order literal, let $\phi^*(R,\tuple x)$ be just $\phi(\tuple x)$. 
Then we have
\begin{align*}
M \models_X \phi(\tuple x) &\iff M \models_s \phi(\tuple x)\text{ for all $s \in X$}\\
&\iff (M,X(\tuple x)) \models_s \phi(\tuple x)\text{ for all $s \in X$}.
\end{align*}

\item If $\phi(\tuple x)$ is an inclusion atom $\tuple t_1 \subseteq \tuple t_2$, let $\phi^*(R,\tuple x)$ be 
$\exists \tuple z(R \tuple z \land \tuple t_1(\tuple x) = \tuple t_2(\tuple z))$,
where $\tuple z$ is a tuple of new variables. Note that 
$(M,X(\tuple x)) \models_h  \tuple t_1(\tuple x) = \tuple t_2(\tuple z)$ for an assignment 
$h$ defined on $\tuple x\tuple z$ if and only if there are two assignments 
$s,s'$ defined on $\tuple x$ such that  $\tuple t_1\langle s\rangle=\tuple t_2\langle s'\rangle$
and  $h= s\cup (s'\circ f)$, where $f$ is the function $f(z_i)=x_i$.
Thus, we see that $(M, X(\tuple x))\models_s \phi^*(R,\tuple x)$ for all $s\in X$
if and only if for every $s\in X$ there is an $s'\in X$ such that 
$\tuple t_1\langle s\rangle=\tuple t_2\langle s'\rangle$, as desired.

\item Assume next that $\phi(\tuple x)$ is of the form $\psi(\tuple x) \lor \theta(\tuple x)$. 
Then we define  
$$
	\phi^*(R,\tuple x):=[\gfp_{S,\tuple x}\,(R\tuple x\land\psi^*(S,\tuple x))]\tuple x\lor 
	[\gfp_{T,\tuple x}\,(R\tuple x\land\theta^*(T,\tuple x))]\tuple x.
$$
If $M\models_X\phi(\tuple x)$, then there are teams $Y$ and $Z$ such that $X=Y\cup Z$,
$M\models_Y\psi(\tuple x)$ and $M\models_Z\phi(\tuple x)$. By induction hypothesis,
$(M,Y(\tuple x))\models_s\psi^*(S,\tuple x)$, and consequently
$(M,X(\tuple x),Y(\tuple x))\models_s R\tuple x\land\psi^*(S,\tuple x)$, holds for all $s\in Y$. 
Hence, by Lemma \ref{gfp-lemma},
$(M,X(\tuple x))\models_s [\gfp_{S,\tuple x}\,(R\tuple x\land\psi^*(S,\tuple x))]\tuple x$
holds for all $s\in Y$.

In the same way we see that
$(M,X(\tuple x))\models_s [\gfp_{T,\tuple x}\,(R\tuple x\land\theta^*(T,\tuple x))]\tuple x$
holds for all $s\in Z$. Thus, we conclude that $(M,X(\tuple x))\models_s \phi^*(R,\tuple x)$
for all $s\in X$.

To prove the converse, assume that $(M,X(\tuple x))\models_s \phi^*(R,\tuple x)$ for all $s\in X$.
Let $Y$ be the set of  all assignments $s\in X$ that satisfy the first disjunct of $\phi^*(R,\tuple x)$,
and let $Z$ be the set of assignments $s\in X$ that satisfy the second disjunct. Then  
$Y$ is the maximal team such that, for all $s\in Y$,
$(M,X(\tuple x))\models_s[\gfp_{S,\tuple x}\,(R\tuple x\land\psi^*(S,\tuple x))]\tuple x$. 
It follows from Lemma~\ref{gfp-lemma} that
$(M,X(\tuple x),Y(\tuple x))\models_s R\tuple x\land \psi^*(S,\tuple x)$ for all $s\in Y$.
Thus, $(M,Y(\tuple x))\models_s \psi^*(S,\tuple x)$ for all $s\in Y$, and by induction hypothesis,
$M\models_Y\psi(\tuple x)$. In the same way we see that $M\models_Z\theta(\tuple x)$.
Finally, since $X=Y\cup Z$, we conclude that $M\models_X\phi(\tuple x)$.

\item If $\phi(\tuple x)=\psi(\tuple x) \land \theta(\tuple x)$, we define simply
$\phi^*(R,\tuple x):=\psi^*(R,\tuple x)\land \theta^*(R,\tuple x)$. The claim
follows then directly from the induction hypothesis.

\item If $\phi(\tuple x)$ is of the form $\exists v \,\psi(\tuple x v)$, let $\phi^*(R,\tuple x)$ be 
\[\exists v [\gfp_{S,\tuple x v}\,(R\tuple x\land\psi^*(S,\tuple x v))]\tuple x v\].
Then $M\models_X \phi(\tuple x)$ if and only if there is a function 
$H\in\CFMX$ such that $M\models_Y\psi(\tuple x v)$,
where $Y=X[H/v]$. By the induction hypothesis, this is equivalent to
$(M,Y(\tuple x v))\models_h \psi^*(S,\tuple x v)$ being true for all $h\in Y$.
This, in turn, is equivalent with the condition
\begin{equation}\label{excond}
	(M,X(\tuple x),Y(\tuple x v))\models_h R\tuple x\land\psi^*(S,\tuple x v)\;
	\text{ for all $h\in Y$}.
\end{equation}
If condition (\ref{excond}) holds, then by Lemma \ref{gfp-lemma},
$(M,X(\tuple x))\models_h[\gfp_{S,\tuple x v}\,(R\tuple x\land\psi^*(S,\tuple x v))]\tuple x v$
holds for all $h\in Y$. Since every $s\in X$ has an extension $h\in Y$, it follows that
$(M,X(\tuple x))\models_s\phi^*(R,\tuple x)$ for all $s\in X$.

On the other hand, if $(M,X(\tuple x))\models_s\phi^*(R,\tuple x)$ for all $s\in X$, we
define $H\in\CFMX$ to be the function such that
\[
	H(s):=\{a\in \dom(M): (M,X(\tuple x))\models_{s[a/v]}[\gfp_{S,\tuple x v}\,
	(R\tuple x\land\psi^*(S,\tuple x v))]\tuple x v\},
\]
and let $Y=X[H/v]$.
Then $Y$ is the maximal team such that
\[(M,X(\tuple x))\models_{h}[\gfp_{S,\tuple x v}\,(R\tuple x\land\psi^*(S,\tuple x v))]\tuple x v\]
for all $h\in Y$, whence condition (\ref{excond}) follows from Lemma \ref{gfp-lemma}.

\item If $\phi(\tuple x)$ is of the form $\forall v \,\psi(\tuple x v)$, let $\phi^*(R,\tuple x)$ be 
\[\forall v [\gfp_{S,\tuple x v}\,(R\tuple x\land\psi^*(S,\tuple x v))](\tuple x v).\] The proof of the claim
is similar to the case of existential quantification.

\end{enumerate}
\end{proof}


In proving that $\PGFP$-sentences can be expressed in $\INCL$
we will use the normal form given in Theorem \ref{gfp-nf}. Thus, it suffices to find 
translations for first-order formulas,
and formulas obtained by a single application of the $\gfp$-operator to first-order formulas.

\begin{theorem}\label{gfp-inc}
Let $\eta(R,\tuple x,\tuple y)$ be a first-order formula 
such that $R$ occurs only positively in $\eta$, $\ar(R)=|\tuple x|=n$, and the free variables
of $\eta$ are in $\tuple x\tuple y$. 
\begin{description}
\item[(a)] There exists an 
$\INCL$-formula $\eta^+(\tuple x,\tuple y)$ such that for all models $M$ and teams $X$ on $M$
\[
	M \models_X \eta^+(\tuple x,\tuple y) \iff  
	(M,X(\tuple x))\models_s\eta(R,\tuple x,\tuple y)
	\text{ for every }s\in X
\]

\item[(b)] If $\tuple y$ is empty, and $\tuple z$ is an $n$-tuple of variables not occurring in $\eta$,
then there exists an $\INCL$-formula $\tilde\eta(\tuple z)$ such that for all models $M$ and 
teams $X$ on $M$
\[
	M \models_X \tilde\eta(\tuple z) \iff  
	M\models_s [\gfp_{R,\tuple x}\,\eta(R,\tuple x)]\tuple z
	\text{ for every }s\in X
\]
\end{description}
\end{theorem}
\begin{proof}
(a) We prove the claim by structural induction on $\eta$. 
\begin{enumerate}

\item If $\eta(R, \tuple x,\tuple y)$ is a first-order literal not containing the relation symbol $R$, 
we define $\eta^+:=\eta$. 
Then $M\models_X\eta^+$ if and only if $M\models_s\eta$ for every $s\in X$. Since $R$ does
not occur in $\eta$, this is equivalent with $(M, X(\tuple x))\models_s \eta$ for all $s\in X$, as 
required.

\item If $\eta$ is of the form $R \tuple t$, we define $\eta^+(\tuple x,\tuple y):={\tuple t \subseteq \tuple x}$.
Then we have
\begin{align*}
M \models_X \eta^+(\tuple x,\tuple y) &\iff \forall s\in X\,\exists s'\in X:
 \tuple t\langle s\rangle=\tuple x\langle s'\rangle\\
&\iff\forall s\in X: \tuple t\langle s\rangle\in X(\tuple x)\\
&\iff \forall s\in X: (M,X(\tuple x)) \models_s R\tuple t.
\end{align*}

\item If $\eta$ is of the form $\alpha(R, \tuple x,\tuple y) \lor \beta(R, \tuple x,\tuple y)$, 
let $\tuple u=(u_1,\ldots,u_n)$ be a tuple of new variables and let $\eta^+(\tuple x,\tuple y)$ be the formula
\[
	 \exists \tuple u \Big( (\tuple u \subseteq \tuple x) \land (\alpha^+(\tuple u, \tuple x \tuple y) \lor \beta^+(\tuple u, \tuple x \tuple y)) \Big).
\]
Here we assume as induction hypothesis that $M \models_Y \alpha^+(\tuple u, \tuple x\tuple y)$ if and only if  $(M, Y(\tuple u)) \models_h \alpha(R, \tuple x, \tuple y)$ for all $h \in Y$, and similarly for
$\beta^+(\tuple u, \tuple x\tuple y)$ and $\beta(R, \tuple x, \tuple y)$.


Suppose first that $M \models_X \eta^+(\tuple x, \tuple y)$. Then there is a function $H \in \CFMXn$ such that $X[H/\tuple u](\tuple u) \subseteq X(\tuple x)$, and furthermore, $X[H/\tuple u]$ can be split into two subteams $Y$ and $Z$ such that $M \models_Y \alpha^+(\tuple u, \tuple x \tuple y)$ and $M \models_Z \beta^+(\tuple u, \tuple x \tuple y)$. Now take any $s \in X$ and let $h\in X[H/\tuple u]$ be an extension of $s$. If $h \in Y$ then $(M, Y(\tuple u)) \models_h \alpha(R, \tuple x, \tuple y)$. Since $Y(\tuple u) \subseteq X[H/\tuple u](\tuple u) \subseteq X(\tuple x)$, $\tuple x \tuple y\langle h\rangle = \tuple x \tuple y\langle s\rangle$ and $R$ occurs only positively in $\alpha$, we have $(M, X(\tuple x)) \models_s \alpha(R, \tuple x, \tuple y)$. Similarly, if $h \in Z$ then $(M, X(\tuple x)) \models_s \beta(R, \tuple x, \tuple y)$. Thus, $(M, X(\tuple x)) \models_s \alpha(R, \tuple x, \tuple y) \lor \beta(R, \tuple x, \tuple y)$ for all $s \in X$, as required.
\smallskip

Conversely, suppose that for any $s \in X$, $(M, X(\tuple x)) \models_s \alpha(R, \tuple x, \tuple y) \lor \beta(R, \tuple x, \tuple y)$. Now let  $H\in\CFMXn$ 
be the function such that $H(s) = X(\tuple x)$ for all $s \in X$. Note first that clearly 
$M\models_{X[H/\tuple u]} \tuple u \subseteq \tuple x$. Let $Y = \{h \in X[H/\tuple u] : (M,X(\tuple x)) \models_h \alpha(R, \tuple x, \tuple y)\}$ and 
$Z = \{h \in X[H/\tuple u] : (M,X(\tuple x)) \models_h \beta(R, \tuple x, \tuple y)\}$. By hypothesis,
$X[H/\tuple u]  = Y \cup Z$. 

If $Y \not = \emptyset$, then $Y(\tuple u) = X[H/\tuple u](\tuple u) = X(\tuple x)$: indeed, if $(M, X(\tuple x)) \models_h \alpha(R, \tuple x, \tuple y)$ then the same holds for all $h'$ which differ from $h$ only with respect to $\tuple u$, since $\tuple u$ is not free in $\alpha$. Therefore $(M, Y(\tuple u)) \models_h \alpha(R, \tuple x, \tuple y)$ for all $h \in Y$, and thus 
$M \models_Y \alpha^+(\tuple u, \tuple x \tuple y)$. 
If instead $Y = \emptyset$, then $M \models_Y \alpha^+(\tuple u, \tuple x \tuple y)$ trivially. 
Similarly, $M \models_Z \beta^+(\tuple u, \tuple x \tuple y)$, and therefore $M \models_{X[H/\tuple u]}  \alpha^+(\tuple u,\tuple x\tuple y)\lor\beta^+(\tuple u,\tuple x\tuple y)$, whence 
the function $H$ witnesses that $M \models_X \eta^+$.
\item If $\eta$ is $\alpha(R, \tuple x,\tuple y) \land \beta(R, \tuple x,\tuple y)$, let
$\eta^+(\tuple x,\tuple y)$ be $\alpha^{+}(\tuple x,\tuple y) \land \beta^+(\tuple x,\tuple y)$.
Then the claim follows directly from the induction hypothesis.

\item If $\eta(R, \tuple x,\tuple y)$ is $\exists v \,\alpha(R, \tuple x, \tuple y v)$, let 
$\eta^+(\tuple x,\tuple y)$ be $\exists v\, \alpha^+(\tuple x,\tuple y v)$; here
we assume w.l.o.g. that $v$ is not among the variables in $\tuple x\tuple y$.  
Then $M \models_X \eta^+(\tuple x,\tuple y)$ if and only if there is a function
 $H\in\CFMX$ such that
$M\models_{X[H/v]}\alpha^+(\tuple x,\tuple y v)$. Since $X[H/v](\tuple x)=X(\tuple x)$, by
induction hypothesis this is equivalent with the condition 
\begin{equation}\label{exeq}
	(M,X(\tuple x))\models_h \alpha(R,\tuple x,\tuple y v)\text{ holds for all $h\in X[H/v]$}.
\end{equation}
If condition (\ref{exeq}) is true, then clearly $(M,X(\tuple x))\models_s \eta(R, \tuple x, \tuple y)$
for all $s\in X$. Conversely, if $(M,X(\tuple x))\models_s \eta(R, \tuple x, \tuple y)$
holds for all $s\in X$, then (\ref{exeq}) is true for the function $H$ such that 
$H(s)=\{a\in\dom(M): (M,X(\tuple x))\models_{s[a/v]}\alpha(R,\tuple x,\tuple y v)\}$.

\item If $\eta(\tuple R, \tuple x,\tuple y)$ is $\forall v\, \alpha(\tuple R, \tuple x,\tuple y v)$, let 
$\eta^+(\tuple x,\tuple y)$ be 
$\forall v\, \alpha^+(\tuple x,\tuple y v)$. The proof of the claim is similar as in the previous case.
\end{enumerate}

(b) Let $\tuple z$ be an $n$-tuple of variables not occurring in $\eta$. 
We define $\tilde\eta(\tuple z)$ to be the formula 
$\exists\tuple x(\tuple z\subseteq\tuple x\land\eta^+(\tuple x))$, where
$\eta^+$ is the $\INCL$-formula corresponding to $\eta(R,\tuple x)$, as given
in claim (a).
Suppose first that $M\models_X\tilde\eta(\tuple z)$. Then there is 
a function $H\in\CFMXn$ such that
$M\models_Y\eta^+(\tuple x)$,
and $\tuple z\langle h\rangle\in Y(\tuple x)$ for all $h\in Y$, where $Y=X[H/\tuple x]$.
Thus, by claim~(a),
$(M,Y(\tuple x))\models_h \eta(R,\tuple x)$
holds for all $h\in Y$. It follows now from Lemma \ref{gfp-lemma} that 
$M\models_h [\gfp_{R,\tuple x}\,\eta(R,\tuple x)]\tuple x$
for all $h\in Y$. Since every $s\in X$ has an extension $h\in Y$, and 
$\tuple z\langle s\rangle=\tuple z\langle h\rangle\in Y(\tuple x)$, we conclude that 
$M\models_s [\gfp_{R,\tuple x}\,\eta(R,\tuple x)]\tuple z$
for all $s\in X$.

To prove the converse, assume that 
$M\models_s [\gfp_{R,\tuple x}\,\eta(R,\tuple x)]\tuple z$ for all $s\in X$.
Let $P$ be the greatest fixed point of the formula $\eta(R,\tuple x)$ 
(with respect to $R$ and $\tuple x$) on the model $M$, and let
$H\in\CFMXn$ be the function such that
$H(s)=P$ for every $s\in X$. Let $Y=X[H/\tuple x]$. Then 
$(M,Y(\tuple x))\models_h \eta(R,\tuple x)$ 
for all $h\in Y$, whence by claim (a), we have
$M\models_Y\eta^+(\tuple x)$. Moreover, 
$\tuple z\langle h\rangle\in Y(\tuple x)=P$ for all $h\in H$, whence
$M\models_Y \tuple z\subseteq \tuple x$. Thus, the function $H$ witnesses that
$M\models_X\exists\tuple x(\tuple z\subseteq\tuple x\land\eta^+(\tuple x))$.
\end{proof}

Note that in the case of disjunction above, it was necessary to ``store'' the possible 
values of $\tuple x$ into the values of a new tuple $\tuple u$ of variables: otherwise, 
by splitting the team $X$ into two subteams we could have lost information about $X(\tuple x)$.

The equivalence of $\INCL$ and $\PGFP$ for sentences follows now from the two theorems above: 

\begin{corollary}\label{incgfpeq}
For any $\INCL$-sentence $\phi$ there exists an equivalent $\PGFP$-sentence $\theta$, and vice versa.
\end{corollary}
\begin{proof}
If $\phi$ is an $\INCL$-sentence, then by Theorem \ref{inc-gfp}, there is a formula
$\phi^*(R,x)$ such that for all models $M$ and teams $X$,
$M \models_X \phi$ if and only if  $(M,X(x)) \models_s \phi^*(R, x)$ for all $s \in X$. Thus,
$M\models\phi$ if and only if $M\models\forall x\, [\gfp_{R,x}\,\phi^*(R, x)]x$.

On the other hand, if $\psi$ is a $\PGFP$-sentence, then by Theorem \ref{gfp-nf},
we can assume that it is of the form $\exists \tuple z\, [\gfp_{R,\tuple x}\,\eta(R,\tuple x)]\tuple z$, 
where $\eta$ is a first-order formula. It follows now from
Theorem \ref{gfp-inc}(b) that $\psi$ is equivalent to the $\INCL$-sentence
$\exists\tuple z\,\tilde\eta(\tuple z)$.
\end{proof}

\begin{corollary}
A class of linearly ordered finite models is definable in $\INCL$ if and only if it can be 
recognized in $\PTIME$.
\end{corollary}

This connection between Inclusion Logic, Fixed Point Logic and descriptive complexity may be of 
great value for the further development of the area. In particular, it implies that fragments and 
extensions of $\INCL$ can be made to correspond to various fragments and extensions of 
$\PTIME$. Hence, results concerning their relationships may lead to insights which may be 
valuable in complexity theory, and vice versa.

\section{First-Order Union Closed Properties}
\label{sub:FOUC}

From Corollary \ref{incgfpeq} it follows immediately that Inclusion Logic is strictly weaker than 
$\Sigma_1^1$. As an immediate consequence, not all $\Sigma_1^1$-definable union-closed properties of relations can be expressed in Inclusion Logic. For example, consider the atom
\begin{description}
\item[TS-$\mathcal R$:] $M \models \mathcal R(xyzw)$ if and only if there exist two functions \[f, g: \dom(M) \rightarrow \dom(M)\] such that, for all $a, b \in \dom(M)$, \[(a, f(a), b, g(b)) \in X(xyzw).\]
\end{description}

It is easy to see that the atom $\mathcal R$ is union-closed.
On the other hand, it can be seen that that the sentence 
$\forall x \exists y \forall z \exists w (\mathcal R(xyzw) \wedge  (x = z \leftrightarrow y = w) \wedge (y=z \rightarrow x = w)\land x\not= y)$ holds in a finite model if and only if it contains an even number of elements.
Since even cardinality is not definable in $\GFP$, it follows that $\mathcal R$ is not definable in $\INCL$. 

But what about first order definable union-closed properties? As we will now see, all such properties are indeed definable in Inclusion Logic; and therefore, it is not possible to increase the expressive power of Inclusion Logic by adding any first order definable union-closed dependency. 
\begin{definition}
A sentence $\phi(R)$ is \emph{myopic} if it is of the form 
$
	\forall \tuple x (R \tuple x \rightarrow \theta(R, \tuple x))
$
for some first-order formula $\theta$ in which $R$ occurs only positively. 
\end{definition}

It follows at once from Theorem \ref{gfp-inc} that myopic sentences correspond to Inclusion Logic-definable properties:
\begin{proposition}
Let $\phi(R) = \forall \tuple x (R \tuple x \rightarrow \theta(R, \tuple x))$ be a myopic sentence. Then there exists an $\INCL$-formula $\phi^+(\tuple x)$ such that, for all models $M$ and teams $X$, 
\[
	M \models_X \phi^+(\tuple x) \mbox{ if and only if } (M,X(\tuple x)) \models \phi(R).
\]
\end{proposition}
\begin{proof}
Consider $\theta(R, \tuple x)$: by  Theorem $\ref{gfp-inc}$, there exists an $\INCL$-formula $\theta^+(\tuple x)$ such that for all models $M$ and teams $X$,
\begin{align*}
	M \models_X \theta^+(\tuple x) &\iff \forall s \in X : (M, X(\tuple x)) \models_s\theta(R, \tuple x)\\ 
	 &\iff (M, X(\tuple x)) \models \forall\tuple x (R \tuple x \rightarrow \theta(R, \tuple x)),
\end{align*}
as required. 
\end{proof}

It is also easy to see that all myopic properties are union-closed. We will now prove the converse implication: if $\phi(R)$ is a first order sentence that defines a union-closed property of relations, then it is equivalent to some myopic sentence. From this preservation theorem it will follow at once that all union-closed first-order properties of relations are definable in Inclusion Logic.

First, let us recall some model-theoretic machinery:
\begin{definition}[$\omega$-big models]
A model $A$ of signature $\Sigma$ is $\omega$-big if for all finite tuples $\tuple a$ of elements of it and for all models $(B, \tuple b, S)$ such that $(A, \tuple a) \equiv (B, \tuple b)$ 
there exists a relation $P$ over $A$ such that $(A, \tuple a, P) \equiv (B, \tuple b, S)$. 
\end{definition}
\begin{definition}[$\omega$-saturated models]
A model $A$ is $\omega$-saturated if for every finite set $C$ of elements of $A$, all complete $1$-types over $C$ with respect to $A$ are realized in $A$.
\end{definition}
The proofs of the following model-theoretic results can be found in  \cite{hodges97b}.
\begin{theorem}[\cite{hodges97b}, Theorem 8.2.1]
Let $A$ be a model. Then $A$ has an $\omega$-big elementary extension.
\end{theorem}
\begin{theorem}[\cite{hodges97b}, Lemma 8.3.4]
Let $A$ and $B$ be $\omega$-saturated structures over a finite signature and such that, 
for all sentences $\chi(R)$ in which $R$ occurs only positively, 
\[
	A \models \chi(R)\; \Longrightarrow\; B \models \chi(R).
\]
Then there are elementary substructures $C$ and $D$ of $A$ and $B$ and a bijective homomorphism $f: C \rightarrow D$ which fixes all relation symbols except $R$.
\end{theorem}

\begin{theorem}[Essentially \cite{hodges97b}, Theorem 8.1.2]
Suppose that $A$ is $\omega$-big and $\tuple a$ is a finite tuple of elements. Then $(A, \tuple a)$ is $\omega$-saturated. 
\end{theorem}

Using these results, we can prove our representation theorem: 

\begin{theorem}\label{FOunion}
Let $\phi(R)$ be a first order sentence that defines a union-closed property of $R$. 
Then $\phi$ is equivalent to some myopic sentence. Consequently,
every first-order definable union-closed property of relations is definable in $\INCL$.
\end{theorem}
\begin{proof}
Let $T = \{\phi'(R) : \phi'(R) \mbox{ is myopic}, \phi(R) \models \phi'(R)\}$. If we can show that $T \models \phi(R)$, we are done: indeed, by compactness this implies that $\phi$ is equivalent to a finite conjunction  
$
	\forall \tuple x (R \tuple x \rightarrow \theta_1(R, \tuple x)) \wedge  \ldots \wedge 
	\forall \tuple x(R \tuple x \rightarrow \theta_n(R, \tuple x))
$
of myopic sentences,
which of course is equivalent to 
$
	\forall \tuple x (R \tuple x \rightarrow (\theta_1(R, \tuple x) \wedge \ldots \wedge \theta_n(R, \tuple x)))
$.

So, let $B'$ be a model satisfying $T$, and let $B$ be an $\omega$-big extension of $B'$. We need to show that $B \models \phi(R)$ (and, therefore, $B' \models \phi(R)$).

Now choose an arbitrary tuple $\tuple b$ of elements such that $B \models R \tuple b$, and let $\Gamma$ be the theory 
\[
	\Gamma = \{R \tuple a, \phi(R)\} \cup \{\psi(R, \tuple a) : R \mbox{ only negative in }
	 \psi, B \models \psi(R, \tuple b)\}.
\]
$\Gamma$ is satisfiable: indeed, if it were not then by compactness there would be formulas $\psi_1(R, \tuple x), \ldots ,\psi_n(R, \tuple x)$ in which $R$ occurs only negatively such that 
\[
	\phi(R) \models \forall \tuple x\Big(R \tuple x \rightarrow \bigvee_{1\le i\le n} \lnot \psi_i(R, \tuple x)\Big).
\]
But this is a myopic formula, and therefore it would have to hold in $B$, which is a contradiction since $B \models \psi_i(R, \tuple b)$ for all $1\le i\le n$.

Now let $(A, \tuple a)$ be an $\omega$-saturated model of $\Gamma$. If $R$ occurs only positively in $\chi(R, \tuple x)$ and $A \models \chi(R, \tuple a)$, then $B \models \chi(R, \tuple b)$; otherwise $\lnot \chi(R, \tuple a)$ would be in $\Gamma$.
Furthermore, since $B$ is $\omega$-big,  $(B, \tuple b)$ is $\omega$-saturated.
Thus, there are elementary substructures $(C, \tuple a)$ and $(D, \tuple b)$ of $(A, \tuple a)$ and $(B, \tuple b)$ and a bijective homomorphism $f: C \rightarrow D$ that fixes all relations except $R$. 

Let $S = f(R^{C})$. Then $S \subseteq R^{D}$, since $f$ is an homomorphism; and $f$ is actually an isomorphism between $(C, \tuple a)$ and $(D[S/R], \tuple b)$, since $f$ fixes even $R$ between these two models.
Now, $C \models R \tuple a \wedge \phi(R)$, whence $D \models S \tuple b \wedge \phi(S)$. Furthermore, since $S \subseteq R$ we have that  $D \models \forall \tuple x (S \tuple x \rightarrow R \tuple x)$. 

Now, $(D, \tuple b)$ is an elementary substructure of $(B, \tuple b)$ and $B$ is a $\omega$-big model: therefore, there exists a relation $P$ over $B$ such that 
$
	(D, \tuple b, S) \equiv (B, \tuple b, P)
$.
In particular, this implies that $B \models P \tuple b \wedge \phi(P) \wedge P \subseteq R$: there is a subset of $R^B$ which contains $\tuple b$ and satisfies $\phi$.

But we chose $\tuple b$ as an arbitrary tuple in $R^B$. So we have that $R^B$ is the union of a family of relations $P_{\tuple b}$, where $\tuple b$ ranges over $R^{B}$; and $B \models \phi(P_{\tuple b})$ for all such $\tuple b$. Since $\phi(R)$ is closed under unions, this implies that $B \models \phi(R)$, as required. 
\end{proof}


\section{An EF Game for Inclusion Logic}
\label{sub:EF}
We will now define an Ehrenfeucht-Fra\"iss\'e game for Inclusion Logic.
This game is an obvious variant of the one defined in \cite{vaananen07} for Dependence Logic: 
\begin{definition}
Let $A$ and $B$ be two models over the same signature, let $n \in \mathbb N$, and let $X$ and $Y$ be two teams with the same domain over $A$ and $B$, respectively. Then the two-player game $G_n(A, X, B, Y)$ is defined as follows: 
\begin{enumerate}
\item The initial position $p_0$ is $(X, Y)$; 
\item For each $i \in \{1, \ldots, n\}$, let $p_{i-1}$ be $(X_{i-1}, Y_{i-1})$. Then Spoiler makes a move of one of the following types: 
\begin{description}
\item[Splitting:] Spoiler chooses two teams $X'$, $X''$ such that $X_{i-1} = X' \cup X''$. Then Duplicator chooses two teams $Y'$, $Y''$ such that $Y_{i-1} = Y' \cup Y''$. Then Spoiler chooses whether the next position $p_i$ is $(X', Y')$ or $(X'', Y'')$.
\item[Supplementing:] Spoiler chooses a variable $v$ and a function $H: X_{i-1} \rightarrow \parts(\dom(A)) \backslash \{\emptyset\}$. Then Duplicator chooses a function $K : Y_{i-1} \rightarrow \parts(\dom(B)) \backslash \{\emptyset\}$, and the new position $p_i$ is\\$(X_{i-1}[H/v], Y_{i-1}[K/v])$.
\item[Duplication:] Spoiler chooses a variable $v$. The next position $p_i$ is $(X_{i-1}[A/v], Y_{i-1}[B/v])$.
\end{description}
\item The final position $p_n = (X_n, Y_n)$ is \emph{winning for Spoiler} if and only if there exists a 
formula $\alpha$ which is either a first-order literal, or an inclusion atom, such that $A \models_{X_n} \alpha$, but $B \not \models_{Y_n} \alpha$. Otherwise, the final position is winning for Duplicator.
\end{enumerate}

\end{definition}
The rank of an Inclusion Logic formula is also defined much in the same way as the rank of a Dependence Logic formula:
\begin{definition}
Let $\phi$ be an $\INCL$-formula. Then we define its \emph{rank} $\rank(\phi) \in \mathbb N$ by structural induction on $\phi$, as follows:
\begin{enumerate}
\item If $\phi$ is a first-order literal or an inclusion atom, $\rank(\phi) = 0$; 
\item $\rank(\psi \wedge \theta) = \max(\rank(\psi), \rank(\theta))$; 
\item $\rank(\psi \vee \theta) = \max(\rank(\psi), \rank(\theta)) + 1$; 
\item $\rank(\exists v \psi) = \rank(\forall v \psi) = \rank(\psi) + 1$.
\end{enumerate}
\end{definition}

The next  theorem 
shows that our games behave as required with respect to our notion of rank. Its proof is practically the same as for the EF game for $\DL$ in \cite{vaananen07}.

\begin{theorem}\label{efgthm}
Let $A$ and $B$ be models and $X$ and $Y$ teams on $A$ and $B$. Then 
Duplicator has a winning strategy in $G_n(A, X, B, Y)$ if and only if 
\[
	A \models_X \phi\; \Longrightarrow\; B \models_Y \phi
\]
holds for all $\INCL$-formulas $\phi$ with $\rank(\phi) \leq n$.
\end{theorem}


Due to the equivalence between $\INCL$ and $\PGFP$ we can conclude 
at once that the EF game for Inclusion Logic is also a novel EF game for $\PGFP$, rather different 
in structure from the one introduced in \cite{bosse93}. It may be hoped that this new game 
and its variants could be of some use for studying the expressive power of fixed point logics.

Although the EF game for Inclusion Logic has a clear second order flavour, it is still manageable:
we will next show that Duplicator has a concrete winning strategy, when the models are simple enough.

\begin{proposition}\label{efex}
Let $A = \{1, \ldots, n\}$ and $B = \{1, \ldots, n+1\}$ be two finite models over the empty signature. Then for all $\INCL$-sentences $\phi$ of rank $\leq n$, 
\[
	A \models \phi \;\Longrightarrow\; B \models \phi.
\]
\end{proposition}
\begin{proof}
It suffices to specify a winning strategy for Duplicator in the game $G_n(A, \{\emptyset\}, B, \{\emptyset\})$. Our aim for such a strategy is to preserve the following property for $n$ turns: 
\begin{itemize}
\item If the current position is $(X, Y)$ then 
\begin{equation}\label{wincond}
	Y = \bigcup \{\pi[X] : \pi\in I(A,B)\},
\end{equation}
\end{itemize}
where $I(A,B)$ is the set of all 1-1 functions $A\to B$,
$\pi[X] = \{\pi(s) : s \in X\}$ and $\pi(s)$ denotes the assignment $\pi\circ s$. 

The property (\ref{wincond}) is trivially true for $(\{\emptyset\}, \{\emptyset\})$. Furthermore, as long as (\ref{wincond}) holds, Spoiler does not win. Indeed, if $\alpha$ is a first-order literal such that $A \models_s \alpha$ for all $s \in X$, then, since all $s' \in Y$ are of the form $\pi(s)$ for some $s \in X$ and the signature is empty, we have $B \models_{s'} \alpha$ for all $s' \in Y$. Similarly, suppose that $A \models_X \tuple u \subseteq \tuple w$, and let $s' \in Y$. Then $s' = \pi(s)$ for some $s \in X$ and some $\pi\in I(A,B)$, and there exists a $h \in X$ such that 
$\tuple u\langle s\rangle = \tuple w\langle h\rangle$. But then $\pi(h) \in Y$, and $\tuple w\langle\pi(h)\rangle =\tuple u\langle  \pi(s)\rangle = \tuple u\langle s'\rangle$, as required. 

Thus, we only need to verify that Duplicator can maintain property (\ref{wincond}) for $n$ rounds. Suppose that at round $i < n$ the current position $(X, Y)$ has property (\ref{wincond}), and let us consider the possible moves of Spoiler: 
\begin{description}
\item[Splitting:] Suppose that Spoiler splits $X$ into $X_1$ and $X_2$. Then let Duplicator reply by splitting $Y$ into 
$Y_j = \bigcup \{s' \in Y : \exists \pi\in I(A,B)\exists s \in X_j \mbox{ such that } \pi(s) = s'\}$
for $j\in \{1,2\}$. Then $Y = Y_1 \cup Y_2$, and it is straightforward to check that both possible successors $(X_1, Y_1)$ and $(X_2, Y_2)$ have property (\ref{wincond}).
\item[Supplementing:] Suppose that Spoiler chooses a function $H\in\CFAX$. Then let Duplicator reply with the function $K\in\CFBY$ defined as
\begin{align*}
	K(s') = \bigcup \{&\pi(a) : \exists\pi\in I(A,B)\exists s \in X \mbox{ such that } \pi(s) = 
	s' \mbox{ and }\\ & a \in H(s)\}
\end{align*}
for each $s' \in Y$. We leave it to the reader to verify that the next position $(X[H/v], Y[K/v])$ has property (\ref{wincond}). 
\item[Duplication:] If Spoiler chooses a duplication move, the next position is $(X[M/v], Y[M/v])$. We  check that this new position satisfies property (\ref{wincond}).

Let $s[a/v] \in X[A/v]$ and let $\pi\in I(A,B)$. Since $s \in X$, we have that $\pi(s) \in Y$, and therefore $\pi(s)[b/v] = \pi(s[a/v]) \in Y[B/v]$. 

Conversely, let $s' \in Y$ and let $b$ be any element of $B$. We need to show that $s'[b/v]=\pi(s[a/v])$ for some $\pi\in I(A,B)$, $s \in X$ and $a \in \dom(A)$. 

By induction hypothesis, there exists $\pi\in I(A,B)$ and $s \in X$ such that $\pi(s) = s'$. If $b$ is in the range of $\pi$, then $s'[b/v]=\pi(s[a/v])$, where $a=\pi^{-1}(b)$.
On the other hand, if $b$ is not in the range of $\pi$, then since $i<n$, there is an element $a\in A$
which is not in the range of $s$. Now $s[a/v]\in X[A/v]$, and $s'[b/v]=\pi'(s[a/v])$, where
$\pi'\in I(A,B)$ is a function such that $\pi'(a)=b$ and $\pi'(c)=\pi(c)$ for all $c$ in the range of $s$.

\end{description}
\end{proof}
From Proposition \ref{efex} it immediately follows that \emph{even cardinality} (and other similar cardinality
properties) of finite models is not definable in Inclusion Logic. This, of course, follows already from
the equivalence of $\INCL$ and $\PGFP$, as it is well-known that non-trivial cardinality properties
are not definable in fixed point logics.


\section{Conclusions and Further Work}
In this work, we proved a number of results concerning the expressive power of inclusion Logic. 
We showed that this logic is strictly weaker than $\Sigma_1^1$, and corresponds in fact to Positive Greatest Fixed Point Logic. Furthermore, we showed that all union-closed first-order properties of relations correspond to the satisfaction conditions of Inclusion Logic formulas, and we also defined a new Ehrenfeucht-Fra\"iss\'e game for it. 

Due to the connection between Inclusion Logic and fixed point logics, the study of this formalism may have interesting applications in descriptive complexity theory. In \cite{durkon12}, Durand and Kontinen established some correspondences between fragments of Dependence Logic and fragments of $\NP$; in the same way, one may hope to find correspondences between fragments of Inclusion Logic and fragments of $\PTIME$. 

Furthermore, we may inquire about extensions of Inclusion Logic. For example, is there any natural union-closed dependency notion $\mathbf D$ such that $\mathrm{FO}(\subseteq, \mathbf D)$ defines all $\Sigma_1^1$ union-closed properties of relations? By the results in Section \ref{sub:FOUC}, we know that if this is the case, then $\mathbf D$ is not first-order.

\section{Acknowledgements}
Pietro Galliani was supported by grant 264917 of the Academy of Finland.

We thank Erich Gr\"adel, Miika Hannula, Juha Kontinen and Jouko\\ V\"a\"an\"anen for a number of highly useful suggestions and comments. We especially thank Miika Hannula for pointing out an error in a previous version of the paper. 




%
\end{document}